\documentclass[12pt,english]{amsart}
\usepackage[T1]{fontenc}
\usepackage[latin9]{inputenc}
\usepackage[letterpaper]{geometry}
\usepackage{float}
\usepackage{amsthm}
\usepackage{amstext}
\usepackage[dvips]{graphicx}
\usepackage{amssymb}

\makeatletter
\numberwithin{equation}{section} 
\numberwithin{figure}{section} 
\theoremstyle{plain}

\usepackage{amsmath,amssymb,amsthm,latexsym,graphicx}

\newtheorem{theorem}{Theorem}[section]
\newtheorem{lemma}[theorem]{Lemma}

\newtheorem{proposition}[theorem]{Proposition}
\theoremstyle{definition}
\newtheorem{definition}[theorem]{Definition}
\theoremstyle{remark}
\newtheorem{remark}[theorem]{Remark}

\newcommand{\Graph}{\Gamma}
\newcommand{\G}{\Gamma}

\newcommand{\cH}{\mathcal{H}}

\newcommand{\Vertices}{\mathcal{V}}
\newcommand{\Edges}{\mathcal{E}}

\newcommand{\R}{\mathbb{R}}

\newcommand{\Reals}{\mathbb{R}}

\newcommand\term[1]{{\em #1\/}}

\DeclareMathOperator{\opt}{opt}

\newcommand\rr{\mathbb{R}}

\newcommand\vv{\mathcal{V}}

\newcommand\nn{\mathcal{N}}

\newcommand\cP{\mathcal{P}_{m}\left(\Gamma\right)}
\newcommand\cQ{\mathcal{Q}_{m}\left(\Gamma\right)}

\newcommand\tuple{\left(\varphi_{1},\ldots,\varphi_{\beta}\right)}
\newcommand\ttuple{\left(\tilde{\varphi}_{1},\ldots,\tilde{\varphi}_{\beta}\right)}

\newcommand\tga{\Gamma_{\varphi_{1},\ldots,\varphi_{\beta}}}
\newcommand\ttga{\Gamma_{\tilde{\varphi}_{1},\ldots,\tilde{\varphi}_{\beta}}}
\newcommand\tf{\tilde{\varphi}}

\newcommand\pp{(-\pi,\pi]}
\newcommand\at{\Big|}

\newcommand\cyclomatic{Betti }

\makeatother

\usepackage{babel}

\begin{document}

\title[The number of nodal domains and the stability of graph partitions]{On the
  connection between the number of nodal domains on quantum graphs and
  the stability of graph partitions}

\author{Ram~Band$^{1,\,2}$}

\address{$^{\text{1}}$Department of Physics of Complex Systems, Weizmann
Institute of Science, Rehovot 76100, Israel}

\address{$^{\text{2}}$School of Mathematics, University of Bristol,
  Bristol BS8 1TW, UK}

\author{Gregory~Berkolaiko$^{3}$}

\address{$^{\text{3}}$Dept of Mathematics, Texas A\&M
  University, College Station, TX 77843-3368, USA}

\author{Hillel Raz$^{4}$}

\address{$^{4}$Cardiff School of Mathematics and WIMCS, Cardiff University,
Cardiff CF24 4AG, UK}

\author{Uzy~Smilansky$^{1,\,4}$}

\begin{abstract}
Courant theorem provides an upper bound for the number of nodal
domains of eigenfunctions of a wide class of Laplacian-type operators. In
particular, it holds for generic eigenfunctions of quantum graph. The theorem
stipulates that, after ordering the eigenvalues as a non decreasing
sequence, the number of nodal domains  $\nu_n$ of the $n$-th
eigenfunction satisfies $n\ge \nu_n$. Here, we provide a new
interpretation for the Courant nodal deficiency $d_n = n-\nu_n$ in the case
of quantum graphs. It equals
the Morse index --- at a critical point --- of an energy functional on a
suitably defined space of graph partitions. Thus, the nodal
deficiency assumes a previously unknown and profound meaning- it is
the number of unstable directions in the vicinity of the critical
point corresponding to the $n$-th eigenfunction. To demonstrate this
connection, the space of graph partitions and the energy functional
are defined and the corresponding critical partitions are studied in
detail.

\end{abstract}
\maketitle

\section{Introduction}

Nodal domains are defined as the connected components that remain
after removing the set of points on which the eigenfunction is
zero. They are easy to observe experimentally and their study has a
rich history. For a review of the subject see, for example, the
collection of articles in \cite{Wittenberg2006} or the book
\cite{NodalSets_book} (in preparation) and references therein.  A
cornerstone in the study of nodal domains is Courant's theorem which
states that, after ordering the eigenvalues as a non-decreasing
sequence, the number $\nu_n$ of nodal domains of the $n$-th
eigenfunction is bounded from above by $n$
\cite{Cou23,CourantHilbert_volume1}.  It was later proven by Pleijel
\cite{Ple_cpam56} that for planar problems with Dirichlet boundary
conditions, ``Courant-sharp'' eigenfunctions that satisfy $\nu_{n}=n$
are finite in number, i.e.  extremely rare (see also
\cite{Pol_pams09}).  The sequence of \emph{nodal deficiencies},
$d_{n}=n-\nu_{n}$ is specific to the particular problem, and it was
recently discovered that this sequence encodes information about the
geometry of the manifold, much in the same way as the eigenvalue
spectrum does \cite{BluGnuSmi_prl02,KarSmi_jpa08}. 
Moreover, the information derived from the nodal count tends to
complement the information contained in the spectrum
and in certain cases it was shown that isospectral systems can be
resolved by their different nodal count
sequences\cite{GnuSmiSon_jpa05,BruKlaPuh_jpa07,Kla_jpa09}.  Thus, the
question ``can one \emph{count} the shape of a drum'' turned out to be
a useful paraphrase of Kac's famous question. The results mentioned
above were extended to quantum graphs, where a Pleijel-like theorem
does not hold but the intimate link between the nodal count sequence
and the graph geometry does
exist. \cite{BanShaSmi_jpa06,Ber_cmp08,BanOreSmi_incol08,BanBerSmi_prep10}.

Another link between the spectral and the nodal properties is achieved
by studying partitions of the domain into subdomains which minimize a
certain energy functional
\cite{BucButHen_amsa98,ConTerVer_cvpde05,CafLin_jsc07}.  Recently,
Helffer, Hoffmann-Ostenhof and Terracini \cite{HelHofTer_aihp09}
proved an important result which connects together the notion of a
minimal partition and the Courant bound.  They consider the Dirichlet
problem on a domain $\Omega$.  For a partition $P$ of $\Omega$ into
subdomains $D_j$, $j=1,\ldots,k$, they define the functional
\begin{equation}
  \label{eq:energy_func}
  \Lambda(P) = \max_{1\leq j \leq k} \lambda_1(D_j),
\end{equation}
where $\lambda_1(D_j)$ is the first eigenvalue of the Dirichlet
Laplacian on $D_j$.  The \emph{minimal partition} is defined as the partition on which the minimum of $\Lambda(P)$ over the set of all $k$-partitions is achieved.  A partition is \emph{bipartite} if its subdomains can be labeled with signs $\{+,-\}$ so that neighboring
domains have different signs.  It was shown in \cite{HelHofTer_aihp09} that a minimal partition is bipartite if and only if it corresponds to the eigenfunction that is Courant-sharp.\footnote{The related question on the meaning of the minimal partitions that are not bipartite is discussed in the review \cite{Hel_mjm10}.}

The result of Helffer \emph{et. al.} \cite{HelHofTer_aihp09} is
surprising and even somewhat mysterious.  It raises the natural
questions: Why only Courant-sharp eigenfunctions appear in the
discussion?  What about other eigenfunctions, which are, according to
Pleijel, the overwhelming majority?  In the present work we address
these questions in the context of quantum graphs which are defined in
Section \ref{sec:intro_quantum_graphs}. Quantum graphs have the
advantage of being simple to analyze without losing the complex
spectral features which mark the Laplacian spectra in more general
domains.  With respect to nodal domains, quantum graphs lie between
$d=1$ and higher dimensions.  On planar domains, Pleijel's theorem
implies that the nodal deficiency
is unbounded from above.\footnote{A related conjecture by T.~Hoffmann-Ostenhof that
  $\limsup_{n\to\infty} \nu_n = \infty$ for all domains is still
  unproven.}  
However, it was shown in \cite{Ber_cmp08} that for
quantum graphs the nodal deficiency is bounded from above by 
the number $\beta$ of independent cycles on the graph.  It turns out that this
feature makes graphs a very good model to study minimal partitions.
We review this and other related results in
Section~\ref{sec:intro_nodal_count} below.

The novel element which we introduce here is that the nodal deficiency
equals the Morse index of energy functional (\ref{eq:energy_func}) on
a suitably defined space of graph partitions. More precisely, we
restrict our attention to the so-called \emph{equipartitions} (see
definition \ref {eq:equipart}) on which the functional $\Lambda$
becomes differentiable.  Under not too restrictive assumptions to be
listed below, eigenfunctions are found to correspond to the bipartite
\emph{critical partitions} of the functional $\Lambda$.  Furthermore,
it turns out that the nodal deficiency of an eigenfunction coincides
with the number of unstable directions (Morse index) of the
corresponding critical partition.  In particular, a minimum has Morse
index 0 and therefore corresponds to an eigenfunction of deficiency 0,
i.e.\ a Courant-sharp eigenfunction.  Thus, our work extends the
result of Helffer, Hoffmann-Ostenhof and Terracini
\cite{HelHofTer_aihp09} to graphs and goes beyond it by interpreting
the nodal deficiency in a new way.

In Section \ref{sec:intro_quantum_graphs} we define quantum graphs
and their spectrum, in \ref{sec:intro_nodal_count} we review the nodal
count results on graphs.  The main results of the paper are presented
in Section \ref{sec:results} and proved in subsequent sections.  In
Section~\ref{sec:other_scenarios} we remove some restrictions we
imposed to keep the development simpler and discuss why other restrictions
cannot be removed.

\subsection{Quantum graphs}

\label{sec:intro_quantum_graphs}

In this section we describe the quantum graph which is a metric graph
with a Shr\"odinger-type self-adjoint operator defined on it. Let
$\Graph=(\Vertices,\Edges)$ be a graph with vertices
$\Vertices=\{v_{j}\}$ and edges $\Edges=\{e_{j}\}$.  The sets
$\Vertices$ and $\Edges$ are required to be finite.

We are interested in metric graphs, i.e.\ the edges of $\Graph$ are
$1$-dimensional segments with a positive finite length $L_{e}$.
On the edge $e=(u,v)$ we assign a coordinate, denoted $x_{e}$, which measures
the distance along the edge starting from one of its vertices. A metric
graph becomes quantum after being equipped with an additional structure:
assignment of a self-adjoint differential operator. This operator
will be often called the \emph{Hamiltonian}. In this paper we study
the zeros of the eigenfunctions of the
Schr\"odinger operator
\begin{equation}
  H\ :\ f(x)\mapsto-\frac{d^{2}f}{dx^{2}}+V(x)f(x),
  \label{E:schrod}
\end{equation}
where $x$ is the coordinate along an edge and $V(x)$ is a
\emph{potential}. We will assume that the potential $V(x)$ is bounded
and piecewise continuous.

To complete the definition of the operator we need to specify its
domain.

\begin{definition} \label{D:spaces} We denote by
  $\widetilde{H}^{2}(\G)$ the space \begin{equation*}
  \widetilde{H}^{2}(\G):=\bigoplus_{e\in\Edges}H^{2}(e),\end{equation*} which
  consists of the functions $f$ on $\G$ that on each edge $e$ belong
  to the Sobolev space $H^{2}(e)$. The restriction of $f$ to the edge
  $e$ is denoted by $f_{e}$. The norm in the space
  $\widetilde{H}^{2}(\G)$ is \begin{equation*}
  \|f\|_{\widetilde{H}^{2}(\G)}:=\sum\limits
  _{e\in\Edges}\|f_{e}\|_{H^{2}(e)}^{2}.\end{equation*}
\end{definition}

We assume that the domain of the Hamiltonian is a subspace of the
Sobolev space $\widetilde{H}^{2}(\G)$.  Note that in the definition of
$\widetilde{H}^{2}(\G)$   smoothness is enforced along edges only,
without any vertex conditions  at all.  All vertex
conditions that lead to the operator \eqref{E:schrod} being
self-adjoint have been classified in
\cite{KosSch_jpa99,Har_jpa00,Kuc_wrm04}.  The conditions involve the
values of the functions $f_e$ and their first derivatives at the
vertices, both of which are well defined by the standard Sobolev trace
theorem.  Since the direction is important for the first derivative,
we will henceforth adopt the convention that, at an end-vertex of an
edge $e$, the derivative is calculated \emph{into} the edge and away
from the vertex.

We will only be interested in the so-called extended $\delta$-type
conditions, since they are the only conditions that guarantee
continuity of the eigenfunctions, something that is essential if one
wants to study changes of sign of the said eigenfunctions.

\begin{definition} \label{def:domain} The domain $\cH$ of the operator
  \eqref{E:schrod} consists of the functions
  $f\in\widetilde{H}^{2}(\G)$ such that
  \begin{enumerate}
  \item $f$ is continuous on every vertex $v\in\vv$: \begin{equation*}
    f_{e_{1}}(v)=f_{e_{2}}(v),\end{equation*}
   for all edges $e_{1}$ and $e_{2}$ that have $v$ as an endpoint.
  \item the derivatives of $f$ at each vertex $v$
    satisfy
    \begin{equation}
      \sum_{e\in\Edges_{v}}\frac{df}{dx_{e}}(v)
      = \alpha_{v}f(v),\quad\alpha_{v}\in\Reals,
      \label{eq:delta_deriv}
    \end{equation}
    where $\Edges_{v}$ is the set of edges incident to $v$.
  \end{enumerate}
\end{definition}

Sometimes the condition \eqref{eq:delta_deriv} is written in a more
robust form
\begin{equation}
  \cos\left(\frac{\varphi_{v}}{2}\right)
  \sum_{e\in\Edges_{v}}\frac{df}{dx_{e}}(v)
  = \sin\left(\frac{\varphi_{v}}{2}\right)f(v), \qquad \varphi_v
  \in\pp,
  \label{eq:delta_deriv_tan}
\end{equation}
which is also meaningful for infinite values of
$\alpha_{v}=\tan\left(\frac{\varphi_{v}}{2}\right)$.  Henceforth we
will understand $\alpha_{v}=\infty$ as the Dirichlet condition
$f(v)=0$. The case $\alpha_{v}=0$ is often referred to as the
Neumann-Kirchhoff condition.

The operator \eqref{E:schrod} with the domain $\cH$ is self-adjoint
for any choice of real $\alpha_{v}$ (including $\alpha_{v}=\infty$).
Since we only consider compact graphs, the spectrum is real, discrete
and with no accumulation points. We will slightly abuse notation and
denote by $\sigma(\Gamma)$ the spectrum of an operator $H$ defined on
the graph $\Gamma$.  The vertex conditions will usually be clear from
the context.

The eigenvalues $\lambda\in\sigma(\Gamma)$ satisfy the equation
\begin{equation}
  -\frac{d^{2}f}{dx^{2}}+V(x)f(x)=\lambda f(x).
  \label{eq:eig_eq}
\end{equation}
It can be shown that under the conditions specified above the operator
$H$ is bounded from below \cite{Kuc_eds04}. Thus we can number the
eigenvalues in an ascending order, starting with $\lambda_1$. As the lowest
eigenvalue plays an important role in this paper, we adopt the physical
terminology and call it the \emph{groundstate energy} and its corresponding
eigenfunction, the \emph{groundstate}.

The Hamiltonian can also be discussed in terms of its quadratic
form \cite{Kuc_wrm04},
\begin{equation}
  h[f,f] = \sum_{e}\int|f'(x)|^{2}dx + \sum_{e}\int V(x)|f(x)|^{2}dx +
 \sum_{v}\alpha_{v}|f(v)|^{2}, \qquad f\in \widetilde{H}^{1}(\G).
  \label{eq:quadratic_form}
\end{equation}
As usual, the Dirichlet conditions (if any) are to be introduced
directly into the domain of the form rather than included in the last
sum of \eqref{eq:quadratic_form}.


The eigenvalues of the Hamiltonian can be obtained from the quadratic
form by applying the Rayleigh-Ritz minimax principle, for instance in
the form
\begin{equation}
  \lambda_{n}=\min_{\dim X=n}\
  \max_{f\in X:\,\|f\|=1}h[f,f],
  \label{eq:minimax}
\end{equation}
where the minimum is taken over all $n$-dimensional subspaces of the
domain of the quadratic form.

Finally, we would like to mention that the Neumann-Kirchhoff and
Dirichlet vertex conditions play an important role in this
paper. Dividing an edge into two parts by introducing a new vertex of
degree 2 will have no effect on the spectrum and eigenfunctions if we
impose the Neumann condition at the vertex.  Indeed, if $\alpha_v=0$
and the degree of $v$ is two, equation (\ref{eq:delta_deriv}) implies
that the derivative of a function from the domain of $\cH$ is
continuous across $v$ and the functions from $H^2$-spaces on the
sub-edges match up to form a valid function from $H^2$-space on the
whole edge. On the other hand, imposing the Dirichlet condition is equivalent to
cutting the graph at the given point and imposing Dirichlet conditions
at the two new vertices of degree 1.  We will therefore consider the
introduction of such a Dirichlet vertex as a change to the
topology of the graph (which might result even in a change of the
number of its connected components). Introducing new vertices on a graph
is a key element in the present paper.


\subsection{Nodal count}
\label{sec:intro_nodal_count}


The main purpose of this article is to investigate the structural
properties of nodal domains of the eigenfunctions of a quantum graph.
In this section we define the nodal domains and review some known
results.

Nodal domains are the connected components of a graph from which the
zero points of a given function have been removed.  More precisely, a
\emph{positive (negative) domain\/} with respect to a function $f$ is
a maximal connected subset in $\G$ where $f$ is positive
(correspondingly, negative).  The total number of positive and
negative domains will be called the \textbf{nodal count} of $f$ and
denoted by $\nu(f)$.  We use $\nu_{n}$ as a shorthand for
$\nu(f_{n})$, where $f_{n}$ is the $n$-th eigenfunction of the graph
in question.  The number of internal zeros of the function $f$ will be
denoted by $\mu(f)$ and $\mu_{n}$ is a shorthand for $\mu(f_{n})$.
Throughout the manuscript we will assume that the zeros of the
function in question do not lie on the vertices of the graph.

The two quantities $\mu$ and $\nu$ are closely related, although, due
to the graph topology, the relationship is more complex than on a
line, where $\nu=\mu+1$.  The topology of the graph comes into play
via the \emph{first Betti number} of $\Graph$ (hereafter, simply
``Betti number''),
\begin{equation}
  \beta_{\Gamma}=|\Edges|-|\Vertices|+1.
  \label{eq:cyclom}
\end{equation}
The graph Betti number has several related interpretations.  In particular, it counts
the number of independent cycles in the graph and gives the
minimal number of edges that need to be removed from $\Gamma$ to turn
it into a tree. Correspondingly, $\beta_{\Gamma}=0$ if and only if
$\Gamma$ is a tree graph, namely if and only if any two vertices of
$\Gamma$ are connected by exactly one path.

The graphs considered in this paper are connected.  However, since the
definition of the nodal domains calls for cutting the graph into
several components, it is beneficial to extend equation
(\ref{eq:cyclom}) to disconnected graphs.  In that case,
$\beta_\Gamma$ is the sum of \cyclomatic numbers of the connected
components, leading to
\begin{equation}
  \beta_{\Gamma}=|\Edges|-|\Vertices|+k,
  \label{eq:cyclom_not_conn}
\end{equation}
where $k$ is the number of connected components of $\Gamma$.

Consider a function $f$ which is non-zero on the vertices of $\Gamma$
and has finitely many isolated zeros.  Denote the set of zeros by $P=P(f)$
and denote by $\Gamma\setminus P$ the graph obtained by cutting
$\Gamma$ at points $P$.  Then, by definition of nodal count,
\begin{equation*}
  \beta_{\Gamma\setminus P} = |\Edges_{\Gamma\setminus P}| -
  |\Vertices_{\Gamma\setminus P}| + \nu(f).
\end{equation*}
Since every cut adds 2 new vertices but increases the number of edges
by 1 only, we get
\begin{equation}
  \label{eq:cyclom_partition}
  \beta_{\Gamma\setminus P} = |\Edges_{\Gamma}| -
  |\Vertices_{\Gamma}| -\mu(f) + \nu(f).
\end{equation}
Combining equations (\ref{eq:cyclom_partition}) and (\ref{eq:cyclom})
we obtain
\begin{equation}
  \label{eq:mu_and_nu_exact}
  \nu(f) = \mu(f) + 1 - \left(\beta_{\Gamma}-\beta_{\Gamma\setminus P}\right).
\end{equation}
In particular, one has the bounds
\begin{equation}
  \mu-\beta_{\G}+1\leq \nu \leq \mu+1.
  \label{E:mu_nu}
\end{equation}

We now concentrate on the nodal count of the \emph{eigenfunctions} of
the graph.  According to the well known ODE theorem by Sturm
\cite{Stu_jmpa36,Stu_jmpa36a,Hin_incol05}, the zeros of the $n$-th
eigenfunction of the operator of type \eqref{E:schrod} on an interval
divide the interval into $n$ nodal domains. By contrast, in the
corresponding question in $\Reals^{d}$, $d\geq 2$, only an upper bound
is possible, given by the Courant's nodal line theorem \cite{CourantHilbert_volume1},
$\nu_{n}\leq n$. In a series of papers
\cite{AlO_viniti92,PokPryObe_mz96,GnuSmiWeb_wrm04,Schapotschnikow06,Ber_cmp08},
it was established that a generic eigenfunction of a quantum graph
satisfies both an upper and a lower bound. Namely, let $\lambda_{n}$
be a simple eigenvalue of the Schr\"odinger operator (\ref{E:schrod}),
on a graph $\Gamma$ and its eigenfunction $f^{(n)}$ be non-zero at all
vertices of $\Gamma$. Then
\begin{equation}
  n-\beta_{\G}\leq\nu_{n}\leq n.
  \label{eq:nodal_domains_bound}
\end{equation}
In fact, a simple modification\footnote{Change the first inequality on
  page 811 of the journal version of \cite{Ber_cmp08} to $\nu_T(\psi)
  = \nu_G(\psi) + \ell - \beta_{G\setminus P}$; note that $\ell$
  was denoting $\beta_G$.} in the proof of the lower bound
\cite{Ber_cmp08} improves the bound to
\begin{equation}
  n - \left(\beta_{\G}-\beta_{\Gamma\setminus P}\right) \leq \nu_{n} \leq n.
  \label{eq:nodal_domains_bound_improv}
\end{equation}
Using formula (\ref{eq:mu_and_nu_exact}) we have a similar formula for
the number of zeros,
\begin{equation}
  n-1 \leq \mu_{n}
  \leq n-1+\left(\beta_{\Gamma}-\beta_{\Gamma\setminus P}\right),
  \label{eq:nodal_zeros_bound_improv}
\end{equation}
or a simpler but weaker version
\begin{equation}
  n-1 \leq \mu_{n} \leq n-1+\beta_{\G}.
  \label{eq:nodal_zeros_bound}
\end{equation}

The conditions for the validity of the above inequalities will be
imposed in the present article as well, thus we give them a name.
\begin{definition}
  An eigenfunction $f_{n}$ of a graph $\Gamma$ is called
  \emph{proper} if it is non-zero on vertices of $\Gamma$ and the
  corresponding  eigenvalue $\lambda_n$ is simple.
\end{definition}
Finally we would like to mention that, unlike the $\Reals^{d}$ case,
even the upper bound $\nu_{n}\leq n$ is in general not valid for
improper eigenfunctions on quantum graphs.


\section{The main results}
\label{sec:results}

In the previous section we reviewed the known results on the
number of zeros of the $n$-th eigenfunction of a quantum graph.  The
aim of this paper is to investigate the qualitative features of the
$n$-th zero set.  The question that should be kept in mind is: given a
set of points on the graph, is there an eigenfunction that is zero at
precisely these points?

\begin{definition}
  \label{def:partition}
  Let $\Gamma$ be a quantum graph.
  \begin{enumerate}
  \item A \term{partition vertex} on $\Gamma$ is a new vertex being
    introduced on an edge of $\Gamma$. The partition vertex is called
    \emph{proper} if it is located in the interior of an edge, that is not at
    an existing vertex of $\Gamma$.  Otherwise, we call it an improper
    partition vertex.
  \item An $m$-\term{partition} of $\Gamma$ is a set of $m$ partition
    vertices on the graph. The partition is proper if all of its
    vertices are proper. Otherwise, we call it an improper
    partition. The set of all proper $m$-partitions of $\Gamma$ is
    denoted by $\cP$.
  \end{enumerate}
\end{definition}

\begin{remark}
  An eigenfunction on a generic graph is expected to be non-zero on
  the vertices of the graph (see \cite{Fri_ijm05} for a related result
  in a special case).  Thus improper partitions are not relevant for
  the study of eigenfunctions on a generic graph.  In section
  \ref{sec:other_scenarios} we will discuss some pathological aspects
  of improper partitions, and point out why our restrictions cannot be
  relaxed.  In the rest of the manuscript a partition would always
  mean a proper one.
\end{remark}

\begin{remark}
  The partition $P\in\cP$ should be understood as a candidate for
  the zero set of an eigenfunction.  As mentioned in section
  \ref{sec:intro_quantum_graphs}, imposing Dirichlet vertex conditions
  at the partition vertices of $P$ separates $\Gamma$ into several
  subgraphs, which we denote by $\left\{ \Gamma_{j}\right\}$, and call
 the partition's subgraphs or connected components.  The number of
  partition components is denoted by $\nu(P)$ and is related to the
  number of partition points $\mu(P)\equiv m$ via
  equation~(\ref{eq:mu_and_nu_exact}).  We chose the number of points
  $m$ to act as the \emph{size of the partition} to simplify the
  subsequent notation.  Making the other possible choice, $m:=\nu(P)$,
  would result in only minor changes to the proof and will have almost
  no effect on the final result.  We note that in dimensions higher
  than 1, the ``number of zeros'' concept is no longer available, and the
  ``number of components'' therefore acts as the size of the
  partition.
\end{remark}

In the definition of nodal domains in
section~\ref{sec:intro_nodal_count} we distinguished positive and
negative domains.  If an eigenfunction is proper it must change sign
at every zero, thus two neighboring domains must have different
sign.  This motivates the following definition.

\begin{definition}
  The partition $P\in\cP$ is called \term{bipartite} if there exists a map
  from its subgraphs to a sign, $\left\{ \Gamma_{j}\right\}
 \rightarrow\left\{ +,-\right\} $, such that neighboring subgraphs
  are mapped to different signs.
\end{definition}

\begin{figure}[ht]
  \begin{centering}
    \hfill
    \begin{minipage}[c]{0.3\columnwidth}%
      \begin{center}
        \includegraphics{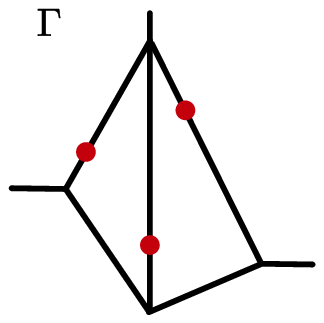}
      \end{center}

      \medskip{}
      \centerline{(a)}
    \end{minipage}
    \hfill{}%
    \begin{minipage}[c]{0.3\columnwidth}%
      \begin{center}
        \includegraphics{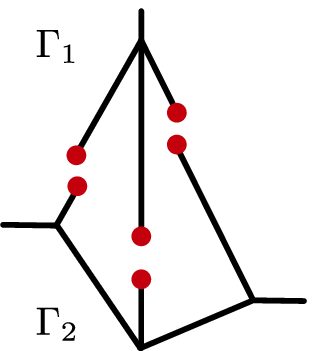}
      \end{center}

      \bigskip\bigskip

      \centerline{(b)}
   \end{minipage}
    \hfill{}
    \begin{minipage}[c]{0.3\columnwidth}%
      \begin{center}
        \includegraphics{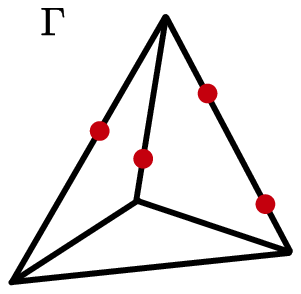}
      \end{center}

      \bigskip\bigskip

      \centerline{(c)}
   \end{minipage}
    \hfill{}
  \end{centering}

  \caption{(a) A proper partition
    $P\in\mathcal{P}_{3}\left(\Gamma\right)$ and (b) its two
    subgraphs; (c) a non-bipartite partition with $\beta_{\Gamma\setminus
      P}=1$}

 \label{fig:partition_and_subgraphs}
\end{figure}

We say that a partition $P$ of $\Gamma$ \term{corresponds} to a
function $f$ on $\Gamma$, or that $f$ \term{corresponds} to $P$ if $f$
vanishes exactly at the partition vertices of $P$.

We aim to characterize the partitions that correspond to the
eigenfunctions of $H$ on $\Gamma$.  As we mentioned already, the
partition must be bipartite.  Further, observe that for a partition
$P$ which corresponds to the $k$-th eigenfunction
$f_{k}$ of $\Gamma$ we have for all $j$
\begin{equation*}
  \lambda_{1}\left(\Gamma_{j}\right) = \lambda_{k}\left(\Gamma\right),
\end{equation*}
where $\lambda_{1}\left(\Gamma_{j}\right)$ is the groundstate
eigenvalue of the Hamiltonian $\left.H\right|{}_{\Gamma_{j}}$
restricted to the $j$-th subgraph $\Gamma_j$ (with Dirichlet
conditions imposed at partition points).  This is because the
restriction of $f_{k}$ to $\Gamma_j$ is an eigenfunction of
$\left.H\right|{}_{\Gamma_{j}}$: it satisfies the eigenvalue
equation~(\ref{eq:eig_eq}) and vanishes at the partition points.  It
must be the groundstate since it does not change sign on $\Gamma_j$.

Thus, for a partition to correspond to an eigenfunction, all
groundstate energies of the partition's subgraphs, $\left\{
  \Gamma_{j}\right\} $, must be equal.  This property is referred to
in the following definition.

\begin{definition}
  An $m$-partition is an \term{equipartition} if all of its subgraphs
  $\left\{ \Gamma_{j}\right\} $ share the same first
  eigenvalue:
  \begin{equation}
    \lambda_{1}(\Gamma_{j_{1}})=\lambda_{1}(\Gamma_{j_{2}})
    \qquad \mbox{for all }j_1,\ j_2.
    \label{eq:equipart}
  \end{equation}
  The set of all proper equipartitions of $\Gamma$ of size $m$ is
  denoted $\cQ$.
\end{definition}

We proceed by defining the following energy functional on $\cP$ (compare
with similar definition in \cite{ConTerVer_cvpde05,HelHofTer_aihp09}):

\begin{definition}
  \label{def:energy_function}
  The functional $\Lambda\,:\,\cP \rightarrow\rr$ is defined by
  \begin{equation*}
   \Lambda\left(P\right) :=\max_{j}\lambda_{1}(\Gamma_{j}).
  \end{equation*}
\end{definition}

The partitions minimizing $\Lambda$ over $\cP$ (defined on
$d$-dimensional domains) were considered in
\cite{ConTerVer_cvpde05,HelHofTer_aihp09}.  However, it is easy to show (as proved
for graphs in the next theorem) that any local minimum of $\Lambda$ on
$\cP$ must be an equipartition.  But first we need the notion of
proximity for partition.  We define the $\varepsilon$-neighborhood of
a partition $P\in\cP$ to be the set of all the partitions obtained by
perturbing the positions of $P$'s partition vertices by a distance
smaller than $\varepsilon$.

\begin{theorem}
  \label{thm:minima_of_Lambda}
  Let $P$ be a local minimum of $\Lambda$ on $\cP$. Then
  $P\in\cQ$.
\end{theorem}

\begin{proof}[Proof of theorem \ref{thm:minima_of_Lambda}]
  We will prove the theorem by contradiction, by showing that any
  partition $P\not\in\cQ$ can be perturbed to decrease the energy
  $\Lambda$.  The perturbation will be performed upon one partition
  point at a time, and will use the fact that elongating the edge
  connected to a degree one vertex with Dirichlet condition decreases
  the groundstate energy.  This follows from the
  well known Hadamard formula for the derivative of an eigenvalue with
  respect to the variation of the domain (see \cite{BerKuc_prep10} for
  the quantum graph adaptation of the formula).

  Let $P\in\cP$ be a local minimum of $\Lambda$. Assume that $P$ is
  not an equipartition. We show that we can perturb the positions of
  the partition vertices of $P$ in a way which decreases the value of
  $\Lambda$ and arrive to a contradiction. Let $\left\{
    \Gamma_{j}\right\} $ be the set of all connected components of
  $P$.  Further, let $\mathcal{I}$ be the set of all $i$ such that
  $\lambda_{1}\left(\Gamma_{i}\right)=\Lambda\left(P\right)$.  In
  particular, since $P$ is not an equipartition, there exist two
 neighboring components, $\Gamma_{i}$ and $\Gamma_{j}$, such that
  $\lambda_{1}\left(\Gamma_{j}\right) <
  \lambda_{1}\left(\Gamma_{i}\right) = \Lambda\left(P\right)$, that is
  $i\in\mathcal{I}$ and $j\notin\mathcal{I}$.  Let $v$ be a partition
  vertex which belongs to the common boundary of $\Gamma_{i}$ and
  $\Gamma_{j}$.  We modify $P$ by slightly moving $v$ into
  $\Gamma_{j}$.  Such a perturbation increases
  $\lambda_{1}\left(\Gamma_{j}\right)$ and decreases
  $\lambda_{1}\left(\Gamma_{i}\right)$, as discussed above.  We use a
  perturbation small enough such that the relation
  $\lambda_{1}\left(\Gamma_{j}\right) <
  \lambda_{1}\left(\Gamma_{i}\right)$ still holds.  After performing
  this perturbation the size of the set $\mathcal{I}$ is reduced by
  one.  We continue perturbing the partition vertices' positions in
  the same manner until we exhaust this set.  Finally, for the
  modified partition $P'$ we have $\Lambda\left(P'\right) <
  \Lambda\left(P\right)$, which contradicts $P$ being a minimum of
  $\Lambda$.
\end{proof}

Recognizing the significance of equipartitions, we wish to further
investigate the energy functional $\Lambda$ restricted to $\cQ$.  In
section \ref{sec:param} we prove the following theorem which 
describes a parameterization of the set of equipartitions $\cQ$.

\begin{theorem}
  \label{thm:parametr}
  Let $\Gamma$ be a finite connected graph with the \cyclomatic number
  $\beta$. Then there exists a number $N$ such that for all $m>N$
  \begin{enumerate}
  \item \label{i:domain_open} there is a map $\Phi_m$ defined on an
    open subset of the torus $\mathbb{T}^\beta = \pp^{\beta}$,
  \item \label{i:bijection} the map acts bijectively between its domain and
    the set of proper equipartitions $\cQ$.
  \item \label{i:smooth} the functional $\Lambda\circ\Phi_m$ is smooth.
  \end{enumerate}
\end{theorem}

\begin{remark}
  In the proof of theorem~\ref{thm:parametr} the map $\Phi_m$ will be
  constructed explicitly.  Furthermore, in
  theorem~\ref{thm:loc_parametr} we will lift the restriction $m>N$ by
  sacrificing the global structure of the map $\Phi_m$.
\end{remark}

We may use theorem \ref{thm:parametr} to allow ourselves from now on
to identify $m$-equipartitions with elements
$\vec{\varphi}=\tuple\in\pp^{\beta}$ without mentioning the map.
In particular, we can consider the energy functional $\Lambda$ to
be defined on the domain of $\Phi_m$.
This allows us to state the main result
of this manuscript:

\begin{theorem}
  \label{thm:critical_point}
  Let $\Gamma$ be a finite connected graph. Let
  $m$ be large enough such that the properties in theorem
  \ref{thm:parametr} hold.
  \begin{enumerate}
  \item \label{enu:critical_point_part1} If a bipartite proper
    equipartition is a critical point of $\Lambda$, then it corresponds
    to an eigenfunction of $\Gamma$. Conversely, the partition which
   corresponds to a proper eigenfunction of $\Gamma$ is a critical
    point of $\Lambda$.
  \item \label{enu:critical_point_part2} If the critical point
    corresponding to the $n$-th eigenfunction is non-degenerate, the
    nodal deficiency $d_{n}=n-\nu_n$ of the eigenfunction is equal to
    the Morse index (the number of unstable directions) of the
    critical point.
  \end{enumerate}
\end{theorem}

\begin{remark}
  Taking 2-dimensional space as an example, a non-degenerate minimum
  has Morse index 0, a saddle point has index 1 and a maximum index
  2.  Thus minima correspond to Courant-sharp eigenfunctions, as
  proved for $\R^d$ domains in \cite{HelHofTer_aihp09}.

  The non-degeneracy assumption is introduced to present the theorem
  in its most elegant form.  In fact, in
  section~\ref{sec:mixed_minimax} we will prove a certain mixed
  minimax characterization of the critical points corresponding to the
  eigenfunction.  The nodal deficiency will be equal to the number of
  maximums taken.  We use the non-degeneracy assumption only to go
  from the minimax to the Morse index.  Bypassing this step it is easy
  to see that even a minimum with a degenerate Hessian still
  corresponds to a Courant-sharp eigenfunction.
\end{remark}

\section{Parameterizing the equipartitions}
\label{sec:param}

In this section we prove Theorem~\ref{thm:parametr} by explicitly
constructing the bijection $\Phi_m$ between an open subset of the torus
$\mathbb{T}^\beta = \pp^\beta$ and the set of proper equipartitions
$\cQ$.

\subsection{Description of the map $\Phi_m$}

Denote
$\beta=\beta\left(\Gamma\right)$ and choose $\beta$ edges of $\Gamma$
such that upon their removal we are left with a tree graph.  Choose a
$\beta$-partition $S$ such that it has one vertex on each of the edges
chosen above. This guarantees $\nu\left(S\right)=1$ and that the
single connected component of $S$ is a tree (figure
\ref{fig:graph_and_resulting_tree}).  Denote the partition vertices by
$\left\{ v_{i}\right\} _{i=1}^{\beta}$ and note that each of them
generates two vertices of degree one.  We denote these vertices by
$\left\{ v_{i}^{-},v_{i}^{+}\right\} _{i=1}^{\beta}$ (according to the
vertex $v_{i}$ of $S$ which generated them) and equip each pair with
the following $\delta$-type conditions:
\begin{equation}
  \begin{split}
    \cos\left(\frac{\varphi_{i}}{2}\right) f'\left(v_{i}^{-}\right)
    & =-\sin\left(\frac{\varphi_{i}}{2}\right) \,
    f\left(v_{i}^{-}\right)
    \\
    \cos\left(\frac{\varphi_{i}}{2}\right) f'\left(v_{i}^{+}\right)
    & =\sin\left(\frac{\varphi_{i}}{2}\right) \,
    f\left(v_{i}^{+}\right),
  \end{split}
  \label{eq:delta_conditions_on_tree}
\end{equation}
for some $\varphi_{i}\in\pp$.  We denote the resulting tree graph by
$\tga$.

\begin{center}
  \begin{figure}[ht]
    \begin{centering}
      \hfill{}%
      \begin{minipage}[c]{0.3\columnwidth}%
        \begin{center}
          \includegraphics{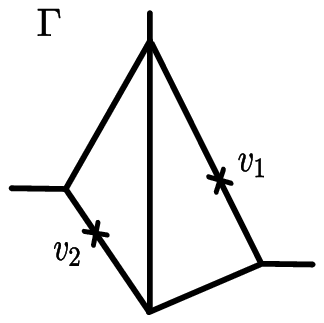}
          \par\end{center}

        \medskip{}

        \begin{center}
          (a)
          \par\end{center}%
      \end{minipage}\hfill{}%
      \begin{minipage}[c]{0.3\columnwidth}%
        \begin{center}
          \includegraphics{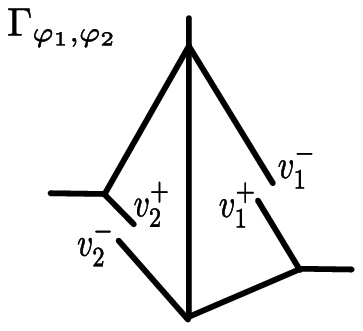}
          \par\end{center}

        \bigskip{}

        \begin{center}
          (b)
          \par\end{center}%
      \end{minipage}\hfill{}
      \par\end{centering}

    \caption{(a) The graph $\Gamma$ with a partition
      $P\in\mathcal{P}_{2}\left(\Gamma\right)$.  ~(b) The resulting tree,
      $\Gamma_{\varphi_{1},\varphi_{2}}$.}

    \label{fig:graph_and_resulting_tree}
  \end{figure}
\end{center}

We describe a parameterization of the space of $m$-equipartitions by a
map $\Phi_{m}$ from $\pp^{\beta}$ to the set of equipartitions
$\cQ$.  The action of this map is as follows.

Let $\tuple\in(-\pi,\pi)^{\beta}$. Examine the eigenspace of the
$(m+1)$-th eigenvalue of the tree $\tga$.  If all of its
eigenfunctions vanish at some vertex (apart from the vertices with the
Dirichlet condition) then the point $\tuple$ is not in the domain of
the definition of the map.

Otherwise, we get by proposition \ref{prop:non_zero_is_simple} that
the $(m+1)$-th eigenvalue is simple.  We can thus apply the nodal
bound \eqref{eq:nodal_zeros_bound} for trees
\cite{PokPryObe_mz96,Schapotschnikow06} (see also \cite{BerKuc_prep10}
for a short proof) and conclude that the $(m+1)$-th eigenfunction has
exactly $m$ zeros on the tree (figure
\ref{fig:tree_and_graph_with_partition_vtcs}(a)).  The location of the
zeros defines a partition $Q\in\cP$ on the original graph $\Gamma$,
see figure \ref{fig:tree_and_graph_with_partition_vtcs}(b).

We now extend the action of the map $\Phi_m$ from $(-\pi,\pi)^{\beta}$
to $\pp^\beta$ by continuity.  Indeed, the relevant eigenvalue is
simple and thus depends analytically on the parameters $\tuple$, see
\cite{BerKuc_prep10}, since $\varphi_j = \pi$ is not particularly
different from any other values of $\varphi_j$ with respect to the
vertex conditions (\ref{eq:delta_conditions_on_tree}).  However, the
varying eigenvalue does not remain the eigenvalue number $m+1$.  In
general we get the $(m-p+1)$-th eigenvalue, where $p$ is the number of
$\varphi_j$ that are equal to $\pi$.  This is because as $\varphi_j
\to \pi$, a zero of the $(m+1)$-th eigenfunction is approaching the
vertex $v_j^-$; at $\varphi_j=\pi$ this zero becomes the boundary
condition at $v_j^-$ and therefore no longer contributes to the nodal
count.

The above extension could have been peformed in the limit $\varphi_j
\to -\pi$ with identical results.  It is thus apparent that $\Phi_m$ is
actually defined on a $\beta$-torus.

\begin{figure}[ht]
\begin{centering}
\hfill{}%
\begin{minipage}[c]{0.3\columnwidth}%
\begin{center}
\includegraphics{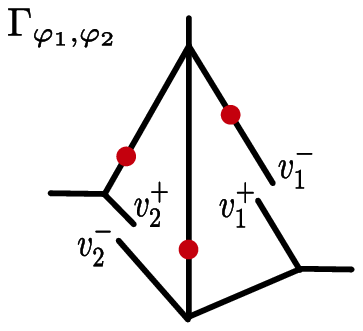}
\par\end{center}

\medskip{}

\begin{center}
(a)
\par\end{center}%
\end{minipage}\hfill{}%
\begin{minipage}[c]{0.3\columnwidth}%
\begin{center}
\includegraphics{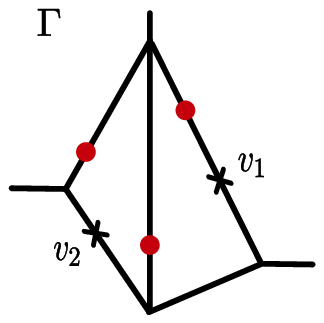}
\par\end{center}

\bigskip{}

\begin{center}
(b)
\par\end{center}%
\end{minipage}\hfill{}
\par\end{centering}

\caption{(a) The tree graph $\Gamma_{\varphi_{1},\varphi_{2}}$ and the
  zeros of its fourth eigenfunction. ~(b) The partition
  $Q=\Phi_{3}\left(\varphi_{1},\varphi_{2}\right)$ on the original
  graph $\Gamma$ (marked with circles).}

\label{fig:tree_and_graph_with_partition_vtcs}
\end{figure}

Let $\tuple$ be in the domain of $\Phi_m$. Then, as we have already
observed in the definition of $\Phi_m$, the $(m+1)$-th eigenvalue of
the tree $\tga$ is simple. Therefore (see \cite{BerKuc_prep10}) it
depends analytically on the vertex conditions, that is on the
parameters $\tuple$. In particular, it remains simple in an open ball
around the initial $\tuple$.  This proves part \ref{i:domain_open} of
Theorem~\ref{thm:parametr}.

In order to distinguish between the points of partitions $S$ and $Q$,
we call the former \emph{section points}. We emphasize that the
location of the section points, $\left\{ v_{i}\right\}
_{i=1}^{\beta}$, is fixed and determines the action of the map
$\Phi_{m}$. The image of the map, $Q=\Phi_{m}\tuple$, gives the other
set of partition vertices and we claim that $Q$ is an
equipartition.

\subsection{The minimal value of $m$}

Next we make precise our requirement on the number $N$ starting from
which the rest of Theorem~\ref{thm:parametr} is guaranteed to
be valid.

\begin{lemma}
  \label{lem:estimateN}
  There exists an $N$ such that for all integers $m>N$ the partition
  $Q=\Phi_m\tuple$ of $\Gamma$ will have $\beta_{\Gamma\setminus Q}=0$
  for every $\tuple$ in the domain of $\Phi_m$.
\end{lemma}

\begin{proof}
  Our proof is constructive: we give an estimate of $N$.  However,
  controlling the \cyclomatic number of a resulting partition is
  hard. Instead we show that for large enough $m$ the partition $Q$ is
  guaranteed at least one point on every edge of $\Gamma$.

  Take an edge $e$ of $\Gamma$ and consider the operator $H$
  restricted to this edge with the Dirichlet conditions at the
  endpoints. Denote the first eigenvalue of $H$ on $e$ by
  $\lambda_{e}$. For example, if the potential $V\equiv0$,
  $\lambda_{e}$ is equal to $\pi/L_{e}$, where $L_{e}$ is the length
  of $e$. Define
  \begin{equation}
    \lambda_D=\max_{e\in\mathcal{E}(\Gamma)}\lambda_{e}.
    \label{eq:Lambda_def}
  \end{equation}
  We are now going to prove that: (i) for large enough $m$ the
  $(m+1)$-th eigenvalue of $\tga$ is larger than $\lambda_D$ for all
  values of $\tuple$ and (ii) if the eigenvalue is larger than
  $\lambda_D$, the corresponding eigenfunction of $\tga$ has a zero on
  every edge of $\Gamma$. These statements combined would finish the
  proof of the lemma.

  To verify statement (i) we observe that, by Theorem
  \ref{thm:glue_many},
  \begin{equation*}
    \lambda_{m-\beta+1}(\Gamma) \leq \lambda_{m+1}\left(\tga\right).
  \end{equation*}
  Therefore we only need to find $N$ such that
  $\lambda_{N-\beta+1}(\Gamma)\geq\lambda_D$ and then (i) is satisfied
  for all $m>N$.

  Before we discuss statement (ii) we note that some of the edges of
  $\Gamma$ are split into two parts in the graph $\tga$ and some care
  should be taken with these edges. Let $f$ be an eigenfunction of
  $\tga$ with an eigenvalue larger than $\lambda_D$. Let $f_{e}$ be the
  restriction of this function to the edge $e$ of $\Gamma$. If the
  edge $e$ contains a section point $v_{i}$ then the function $f_{e}$
  is likely to be discontinuous at $v_{i}$. To fix this we multiply
  the function on the right side of $v_{i}$ by a suitable constant
  (which is $c=f(v_{i}^{-})/f(v_{i}^{+})$). Now we note that due to
  the special structure of
  conditions~(\ref{eq:delta_conditions_on_tree}), the modified
  function $f_{e}$ is not only continuous at $v_{i}$, but is also
  continuously differentiable.

  The obtained function $f_{e}$ satisfies the differential equation
  $Hf_{e}=\lambda f_{e}$ on the edge $e$. It also satisfies homogenous
  boundary conditions at the endpoints $u$ and $v$ of the edge.
  Namely, $f$ satisfies $f'(u)=\alpha_{u}f(u)$ and
  $f'(v)=\alpha_{v}f(v)$, for suitable values of $\alpha_{u}$ and
  $\alpha_{v}$. Since $\lambda$ is greater than the first Dirichlet
  eigenvalue of the edge $e$, the monotonicity of the spectrum with
  respect to the changes of $\alpha$ (see
  Theorem~\ref{thm:interlacing}) implies that $\lambda$ cannot be the
  first eigenvalue of $H$ on the edge $e$, with the boundary
  conditions given above. Therefore $f_{e}$ has at least one zero.
\end{proof}

\begin{remark}
  For the sake of simplicity of the proof we did not pursue the
  sharpest estimates.  To improve them one can, for example, take the
  maximum in the definition of $\Lambda$ over a set of edges, removing
  which turns the graph $\Gamma$ into a tree.
\end{remark}

\subsection{The map $\Phi_m$ produces equipartitions}
\label{sec:map_equi}

Now we take $m$ larger than $N$ from Lemma~\ref{lem:estimateN}.  Let
$f$ be the $(m+1)$-th eigenfunction of $\tga$. We already observed
that $f$ has exactly $m$ nodal points. Hence $f$ corresponds to some
partition $Q\in\cP$ (figure \ref{fig:tree_and_subgraphs}(a)).
Lemma~\ref{lem:estimateN} guarantees that $\beta_{\Gamma\setminus
  Q}=0$ and therefore the subgraphs $\Gamma_{j}$ of the partition are
trees.  We need to prove that the groundstate energy of every
$\Gamma_j$ is equal to the same value
$\lambda_{m+1}\left(\tga\right)$.  We will do it by considering the
restriction of the function $f$ to $\Gamma_{j}$ and modifying it into
the groundstate of $\Gamma_{j}$.

The restriction of the function $f$ to $\Gamma_{j}$ satisfies all the
vertex conditions on $\Gamma_{j}$ and also satisfies the eigenvalue
equation $Hf=\lambda f$ (with
$\lambda=\lambda_{m+1}\left(\tga\right)$) everywhere apart from those
section points that happen to lie on $\Gamma_{j}$.  At these points
the function $f$ is likely to be discontinuous (figure
\ref{fig:tree_and_subgraphs}(b)).

\begin{figure}[ht]
  \begin{centering}
    \hfill{}%
    \begin{minipage}[c]{0.3\columnwidth}%
      \begin{center}
        \includegraphics{figures/tree_with_partition_vtcs_and_section_pts}
        \par\end{center}

      \medskip{}

      \begin{center}
        (a)
        \par\end{center}%
    \end{minipage}\hfill{}%
    \begin{minipage}[c]{0.4\columnwidth}%
      \begin{center}
        \includegraphics{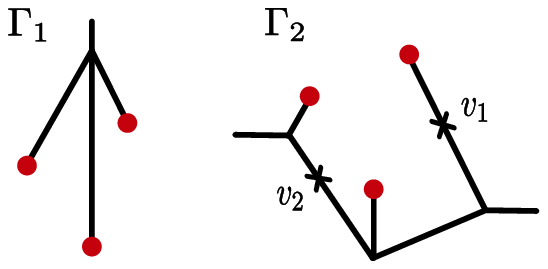}
        \par\end{center}

      \bigskip{}

      \begin{center}
        (b)
        \par\end{center}%
    \end{minipage}\hfill{}
    \par\end{centering}

  \caption{(a) The tree graph $\Gamma_{\varphi_{1},\varphi_{2}}$ and
    the zeros of its fourth eigenfunction, $f$. ~(b) The subgraphs
    of $Q=\Phi_{3}\left(\varphi_{1},\varphi_{2}\right)$ and the
    points $v_{1},v_{2}$ at which $f$ is likely to be
    discontinuous.}

  \label{fig:tree_and_subgraphs}
\end{figure}

We will fix this in a manner similar to the proof of
Lemma~\ref{lem:estimateN}.  Locate a discontinuity point $v_{i}$. The
function satisfies the conditions \eqref{eq:delta_conditions_on_tree}
on the left and right of $v_{i}$. They, in particular, imply that $f$
is not zero at $v_{i}^{\pm}$ . Let
\begin{equation*}
  c=f(v_{i}^{+})/f(v_{i}^{-}).
\end{equation*}
Then, multiplying the function $f$ on the ``left'' part of the tree
$\Gamma_{j}$ (i.e.\ the one connected to $v_{i}^{+}$) by $c$ we make
the resulting function continuous at $v_{i}$. By the special structure
of the conditions \eqref{eq:delta_conditions_on_tree} the new function
is also continuously differentiable at $v_{i}$. Note that we are able
to perform this operation only ``on one side'' of $v_{i}$ (without
affecting $f$ values on the other side) because $\Gamma_{j}$ is a tree
and the vertex $v_{i}$ separates it into two components.

However, multiplication by a constant does not spoil any of the
properties of $f$ at other locations, namely $f$ satisfying vertex
conditions and the eigenvalue equation. By fixing the discontinuities
one by one we arrive at a new function $\tilde{f}$ which has
sufficient regularity properties to be an eigenfunction of the
operator $H$ on $\Gamma_{j}$. Since it has no zeros on $\Gamma_{j}$,
it is the groundstate and therefore
$\lambda=\lambda_{m+1}\left(\tga\right)$ is the first eigenvalue of
$\Gamma_{j}$. We thus obtain that $Q$ is an equipartition and that
\begin{equation}
  \label{eq:LambdaPhi}
  \Lambda\left(Q\right) = \Lambda\left(\Phi_m\tuple\right)
  = \lambda_{m+1}\left(\tga\right).
\end{equation}
As the eigenvalue $\lambda_{m+1}\left(\tga\right)$ is analytic with
respect to the parameters $\varphi_1,\ldots,\varphi_\beta$, we have
also verified part \ref{i:smooth} of Theorem~\ref{thm:parametr}.

\subsection{The map $\Phi_m$ is bijective}

The proof of the bijectivity follows the already established pattern:
from an eigenfunction on one graph (either $\Gamma$ or $\tga$) we
construct an eigenfunction on the other by matching the function in a
smooth way.

We start by remarking that the map $\Phi_m$ is one-to-one. Indeed, in
section~\ref{sec:map_equi}, starting from a point $\tuple$ we
constructed the groundstates on all the connected components, $\left\{
  \Gamma_{j}\right\} $, of the partition $\Phi_m\tuple$.  Suppose that
another point $(\varphi_{1}',\ldots\varphi_{\beta}')$ leads to the
same partition. Then, for every component $\Gamma_{j}$ of the
partition, the same construction leads to a groundstate on
$\Gamma_{j}$. But the groundstate is uniquely determined, up to a
constant, by $\Gamma_{j}$. And for any section point $v_{i}$ that
belongs to $\Gamma_{j}$, the value of $\varphi_{i}$ is uniquely
determined by the corresponding groundstate eigenfunction via
\begin{equation}
  \tan\frac{\varphi_{i}}{2}=\frac{f'(v_{i})}{f(v_{i})}.
  \label{eq:phi_from_f}
\end{equation}
Therefore $\varphi_{i}'=\varphi_{i}$ for every section point $i$.

To show that the map is onto we find, for every equipartition $Q$, the
point $\tuple$ that is mapped to it.  Let $Q\in\cQ$, $m>N$, so that
$\beta_{\Gamma\setminus Q}=0$ by Lemma~\ref{lem:estimateN}.  Let
$\left\{ \Gamma_{j}\right\} $ be the subgraphs of the partition $Q$
and $\left\{ f_{j}\right\} $ their corresponding normalized
groundstates (figure \ref{fig:subgraphs_and_tree}(a)),
\begin{equation*}
  \left.H\right|_{\Gamma_{j}\,}f_{j}
  = \lambda_{1}\left(\Gamma_{j}\right)f_{j}=\Lambda\left(Q\right)f_{j}.
\end{equation*}
As mentioned above, we determine $\varphi_{i}$ from the groundstate of
the partition subgraph $\Gamma_{j}\ni v_{i}$ by
formula~\eqref{eq:phi_from_f}.  To show that the map $\Phi$ sends
$\tuple$ to $Q$ we will construct the $(m+1)$-th eigenfunction of
$\tga$ and verify that its zeros coincide with the partition vertices
of $Q$. Define a function $f$ on $\tga$ by piecing together the
groundstates of the subgraphs $\Gamma_{j}$, i.e.\
$\left.f\right|{}_{\Gamma_{j}}=f_{j}$ (figure
\ref{fig:subgraphs_and_tree}(b)).  This function already goes
considerable distance towards being an eigenfunction of
$\tga$. Indeed, it satisfies the eigenvalue equation $Hf=\lambda f$
and the vertex conditions on $\tga$ at every point except the
partition vertices of $Q$. At the partition vertices $f$ is continuous
(and equal to zero), but might not be differentiable.  We remark that
the function $f$ will satisfy the conditions at vertices $v_{i}^{\pm}$
because of the special way we defined these conditions --- the values
$\tuple$ were especially chosen to fit the function $f$.

\begin{figure}[ht]
  \begin{centering}
    \hfill{}%
    \begin{minipage}[c]{0.4\columnwidth}%
      \begin{center}
        \includegraphics{figures/two_subgraphs_with_section_pts}
        \par\end{center}

      \medskip{}

      \begin{center}
        (a)
        \par\end{center}%
    \end{minipage}\hfill{}%
    \begin{minipage}[c]{0.3\columnwidth}%
      \begin{center}
        \includegraphics{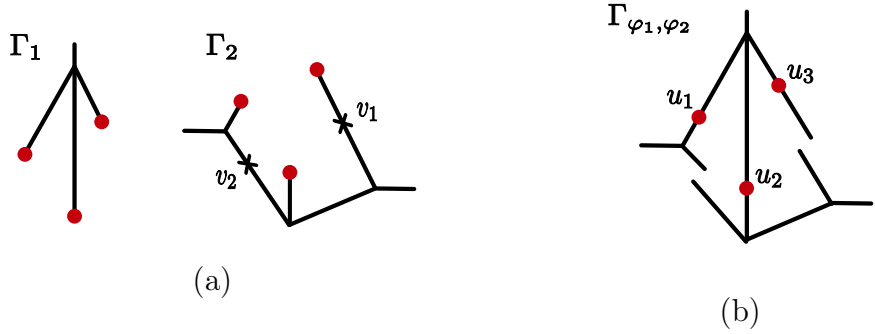}
        \par\end{center}

      \bigskip{}

      \begin{center}
        (b)
        \par\end{center}%
    \end{minipage}\hfill{}
    \par\end{centering}

  \caption{(a) The subgraphs of $Q$, on which the ground states
    $f_{1},f_{2}$ are taken.~(b) The tree graph
    $\Gamma_{\varphi_{1},\varphi_{2}}$ with the points at which the
    function $f$ constructed from piecing together $f_{1},f_{2}$
    might not be differentiable.}

  \label{fig:subgraphs_and_tree}
\end{figure}

We now modify $f$ so that it becomes continuously differentiable at
the partition vertices as well (at the expense of losing the equality
$f(v_i^+)=f(v_i^-)$, a property that we do not need on the tree
$\tga$).  Choosing a partition vertex $u$ we multiply the function $f$
on the right side of it by the suitably chosen constant
\begin{equation*}
  c=f'(u-0)/f'(u+0).\
\end{equation*}
We can perform this operation ``on one side'' of $u$ because $\tga$ is
a tree and the vertex $u$ separates it into two components.
Performing this operation at every partition vertex we will fix all
discontinuities but will not break any other properties of $f$. This
modified $f$ (with a slight abuse of notation) is an eigenfunction on
the tree $\tga$. It is non-zero on all vertices of the tree and
therefore, by Proposition \ref{prop:non_zero_is_simple} its eigenvalue
is simple. It also has exactly $m$ zeros so, by
equation~(\ref{eq:nodal_zeros_bound}), it must be the $(m+1)$-th
eigenfunction of $\tga$. This concludes the proof of
Theorem~\ref{thm:parametr}.

\section{Proof of theorem~\ref{thm:critical_point}}

We begin by recalling that every equipartition $Q\in\cQ$ is an image
of some point $\tuple$ under the map $\Phi_m$. The action of the
map utilizes the $(m+1)$-th eigenfunction of the tree graph $\tga$.
We denote this normalized eigenfunction by $f$ and have
\begin{equation}
  \Lambda\tuple=\lambda_{m+1}\left(\tga\right)=h[f,f],
  \label{eq:Lambda_lambda_and_h}
\end{equation}
where $h[f,f]$ is the quadratic form \eqref{eq:quadratic_form}. In
what follows we will also use the sesqui-linear form,
\begin{equation*}
  h[f,g] = \sum_{e}\int f'(x)\overline{g'(x)}dx
  + \sum_{e} \int V(x)f(x)\overline{g(x)} dx
  + \sum_{v} \alpha_{v} f(v) \overline{g(v)}.
\end{equation*}
The sum in the last term above is over all the vertices of $\tga$
(with the exception of the Dirichlet vertices).  In particular, for the
vertices $\left\{ v_{j}^{-},v_{j}^{+}\right\} _{j=1}^{\beta}$, we have
\begin{equation*}
  \alpha_{j}^{-}  =-\tan\frac{\varphi_{j}}{2} \qquad \qquad
  \alpha_{j}^{+}  =\tan\frac{\varphi_{j}}{2},
\end{equation*}
by the definition of the tree $\tga$, see
\eqref{eq:delta_conditions_on_tree}.

\subsection{Critical points and eigenfunctions}

A critical point is a point where the gradient of
$\Lambda\left(\vec{\varphi}\right)$ is equal to zero.  We
differentiate $\Lambda$ using \eqref{eq:Lambda_lambda_and_h},
\begin{equation}
  \frac{\partial}{\partial\varphi_{j}}\Lambda
  = \frac{\partial}{\partial\varphi_{j}} \left( h\left[f,f\right]
  \right)
  = h_{\varphi_{j}}\left[f,f\right]
  + h\left[ f_{\varphi_j}, f\right]
  + h\left[f, f_{\varphi_{j}}\right],
  \label{eq:critical_point}
\end{equation}
where the subscript $\cdot_{\varphi_{j}}$ stands for the partial
derivative with respect to $\varphi_{j}$.

We now show that the last two terms in the right-hand side of
\eqref{eq:critical_point} vanish.  Recall that $f$ denotes the
normalized $(m+1)$-th eigenfunction of $\tga$.  From the normalization
of $f$ we get
\begin{equation*}
  \frac{\partial}{\partial\varphi_{j}}
  \left\langle f,f\right\rangle = 0
  \qquad\Rightarrow\qquad
  \left\langle f_{\varphi_{j}},f\right\rangle =0.
\end{equation*}
On the other hand, $f$ is an eigenfunction, therefore (see
\cite{BerKuc_prep10} for details)
\begin{equation}
  h\left[f_{\varphi_{j}},f\right] =
  \left\langle f_{\varphi_{j}},Hf\right\rangle
  =\lambda
  \left\langle f_{\varphi_{j}},f \right\rangle = 0.
  \label{eq:vanishing_term1}
\end{equation}
Since $H$ is self-adjoint, we also have
$h\left[f,f_{\varphi_{j}}\right]=0$.

Equation \eqref{eq:critical_point} now reduces to
\begin{align}
  \label{eq:quadratic_form_derivative}
  \frac{\partial}{\partial\varphi_{j}}\Lambda
  &= \frac{1}{2\cos^{2}\left( \frac{\varphi_{j}}{2} \right)}
  \left( \left| f\left(v_{j}^{+}\right) \right|^{2}
    - \left| f\left(v_{j}^{-}\right) \right|^{2}\right) \\
  &= \frac{1}{2\sin^{2}\left( \frac{\varphi_{j}}{2} \right)}
  \left( \left| f'\left(v_{j}^{+}\right) \right|^{2}
    - \left| f'\left(v_{j}^{-}\right) \right|^{2}\right),
  \label{eq:quadratic_form_derivative_Dir}
\end{align}
where, in the last step we used the boundary conditions at $v_j^\pm$,
(\ref{eq:delta_conditions_on_tree}).  The last expression is useful
when $\cos(\varphi_j/2)=0$, but in all other cases we will use
\eqref{eq:quadratic_form_derivative}.

Now let $\tilde{Q}\in\cQ$ be a bipartite proper equipartition which
is a critical point of $\Lambda$.  Let $\vec{\varphi}=\ttuple$ be
the point which is mapped to $\tilde{Q}$ and $\tilde{f}$ the
corresponding normalized $(m+1)$-th eigenfunction on the tree
$\Gamma_{\tf_{1},\ldots,\tf_{\beta}}$.
Assume now that $\cos(\varphi_j/2)\neq0$.  The condition that
$\nabla \Lambda = 0$ at $f=\tilde{f}$ implies, via equation
\eqref{eq:quadratic_form_derivative}, that
\begin{equation*}
  \left|\tilde{f}\left(v_{j}^{+}\right)\right|=\left|\tilde{f}\left(v_{j}^{-}\right)\right|
\end{equation*}
for all $j$.

As $\tilde{Q}$ is bipartite, there is an even number of partition
points on every cycle.  Since the function $\tilde{f}$ changes sign
at every partition point, following any cycle from $v_j^+$ to
$v_j^-$ we deduce that the signs of
$\tilde{f}\left(v_{j}^{+}\right)$ and
$\tilde{f}\left(v_{j}^{-}\right)$ must agree.  Therefore
\begin{equation}
  \label{eq:f_tilde_cont_diff}
  \tilde{f}\left(v_{j}^{+}\right) = \tilde{f}\left(v_{j}^{-}\right)
  \qquad \mbox{ and } \qquad
  \tilde{f}'\left(v_{j}^{+}\right) = -\tilde{f}'\left(v_{j}^{-}\right),
\end{equation}
where we used conditions \eqref{eq:delta_conditions_on_tree} to
deduce the second equality from the first.  Note that the
derivatives are taken in the direction away from the vertices
$v_j^{\pm}$, therefore the function $\tilde{f}$, if considered on
the original graph $\Gamma$ is both continuous and continuously
differentiable at all section points $v_j$.  Since $\tilde{f}$
satisfies the eigenvalue equation, it is an eigenfunction of
$\Gamma$.

Note that if $\varphi_j=-\pi$ for some $j$, we can similarly deduce
equation \eqref{eq:f_tilde_cont_diff} starting with
\eqref{eq:quadratic_form_derivative_Dir} and again using
bipartiteness.

To prove the second direction of the statement, we start with an
eigenfunction of $\Gamma$.  It induces an equipartition on $\Gamma$
and the corresponding values of $\tf_j$ can be read off
equation~(\ref{eq:delta_conditions_on_tree}) (see also
equation~(\ref{eq:phi_from_f})).  To prove that the point $\ttuple$
is critical, we note that $\tilde{f}$ is smooth at the section
points, equation \eqref{eq:f_tilde_cont_diff} is obviously satisfied
and \eqref{eq:quadratic_form_derivative} implies that the gradient
is zero.

\subsection{A mixed minimax}
\label{sec:mixed_minimax}

Let $\tilde{f}$ be the $n$-th eigenfunction on $\Gamma$ with the
eigenvalue $\tilde{\lambda}$.  Let $\tilde{Q} \leftrightarrow
\vec{\varphi}=\ttuple$ be the corresponding $m$-equipartition which,
by part \ref{enu:critical_point_part1} is a critical point.

Assume for now that $\tilde{f}$ is non-zero at the section points.
Then $\tilde{f}$ is the $(m+1)$-th eigenfunction of $\ttga$, i.e.,
$\tilde{\lambda}=\lambda_{m+1}\left(\ttga\right)$. We now apply
\eqref{eq:our_interlacing_join} to get
\begin{equation}
  \lambda_{m}\left(\Gamma_{\tf_{2},\ldots\tf_{\beta}}\right)
  \leq
  \lambda_{m+1}\left(\Gamma_{\varphi_{1},\tf_{2},\ldots,\tf_{\beta}}\right)
  \leq
  \lambda_{m+1}\left(\Gamma_{\tf_{2},\ldots\tf_{\beta}}\right),
  \label{eq:interlacing_relation}
\end{equation}
where $\Gamma_{\tf_{2},\ldots\tf_{\beta}}$
is the graph obtained from
$\Gamma_{\varphi_{1},\tf_{2},\ldots,\tf_{\beta}}$
by gluing the vertices $v_{1}^{-}$ and $v_{1}^{+}$ together into a
single vertex $v_{1}$.

The inequalities above hold for all values $\varphi_{1}\in\pp$.  In
addition, since we know that $\tilde{f}$ is an eigenfunction of
$\Gamma$, it is also an eigenfunction of
$\Gamma_{\tf_{2},\ldots\tf_{\beta}}$.
Therefore, when $\varphi_{1}=\tf_{1}$ one of the
inequalities of \eqref{eq:interlacing_relation} should become an
equality.  Namely, there exists some $\sigma_{1}\in\left\{
  0,1\right\}$, such that
\begin{equation}
  \lambda_{m+1-\sigma_{1}}\left(
    \Gamma_{\tf_{2},\ldots\tf_{\beta}}\right)
  = \lambda_{m+1}\left(
    \Gamma_{\tf_{1},\tf_{2},\ldots,\tf_{\beta}}\right).
  \label{eq:eigenvalue_equality1}
\end{equation}
Carrying the last argument by induction we get that
\begin{equation*}
  \lambda_{m+1-\sigma_{1}-\sigma_{2}\ldots-\sigma_{\beta}}\left(\Gamma\right)
  =\lambda_{m+1}\left(\Gamma_{\tf_{1},\tf_{2},\ldots,\tf_{\beta}}\right),
\end{equation*}
for some values
$\sigma_{1},\sigma_{2},\ldots,\sigma_{\beta} \in \left\{ 0,1\right\}$.
We now wish to characterize these values.

Consider first the case $\sigma_1=0$.  Recalling that inequality
\eqref{eq:interlacing_relation} holds for all values
$\varphi_{1}\in\pp$ and combining it with
\eqref{eq:eigenvalue_equality1} allows one to deduce
\begin{equation*}
  \label{eq:max_case}
  \lambda_{m+1}\left(\Gamma_{\tf_{2},\ldots\tf_{\beta}}\right) =
  \max_{\varphi_{1}\in\pp}
  \lambda_{m+1}\left(\Gamma_{\varphi_{1},\tf_{2},\ldots,\tf_{\beta}}\right)
  \qquad \mbox{ if }\sigma_1=0.
\end{equation*}
On the other hand, we similarly get
\begin{equation*}
  \label{eq:min_case}
  \lambda_{m}\left(\Gamma_{\tf_{2},\ldots\tf_{\beta}}\right) =
  \min_{\varphi_{1}\in\pp}
  \lambda_{m+1}\left(\Gamma_{\varphi_{1},\tf_{2},\ldots,\tf_{\beta}}\right)
  \qquad \mbox{ if }\sigma_1=1.
\end{equation*}
Introducing the notation $\opt_{1}=\min$ and $\opt_{0}=\max$ we can
write both equations as
\begin{equation}
  \lambda_{m+1-\sigma_{1}}\left(\Gamma_{\tf_{2},\ldots\tf_{\beta}}\right)
  = \underset{\varphi_{1}\in\pp}{\opt_{\sigma_{1}}}
  \lambda_{m+1}\left(\Gamma_{\varphi_{1},\tf_{2},\ldots,\tf_{\beta}}\right).
  \label{eq:first_optimization}
\end{equation}
The same reasoning gives
\begin{equation}
  \lambda_{m+1-\sigma_{1}-\sigma_{2}}\left(\Gamma_{\tf_{3},\ldots\tf_{\beta}}\right)
  = \underset{\varphi_{2}\in\pp}{\opt_{\sigma_{2}}}
  \lambda_{m+1-\sigma_{1}}\left(\Gamma_{\varphi_{2},\tf_{3}\ldots,\tf_{\beta}}\right).
  \label{eq:second_optimization}
\end{equation}

Next we observe that \eqref{eq:first_optimization} holds not only
for $\varphi_{2}=\tf_{2}$ but also in some neighborhood $I_{2}$ of
$\tf_{2}$.  In fact, it holds for all values of $\varphi_{2}$ if the
value of $\sigma_1$ is allowed to depend on $\varphi_{2}$, but we
would like to keep it constant.  This allows us to substitute
\eqref{eq:first_optimization} into the right-hand side of
\eqref{eq:second_optimization} and obtain
\begin{equation}
  \lambda_{m+1-\sigma_{1}-\sigma_{2}}\left(\Gamma_{\tf_{3},\ldots\tf_{\beta}}\right)
  = \underset{\varphi_{2}\in I_{2}}{\opt_{\sigma_{2}}}
  \underset{\varphi_{1}\in\pp}{\opt_{\sigma_{1}}}
  \lambda_{m+1}\left(\Gamma_{\varphi_{1},\varphi_{2},\tf_{3}\ldots,\tf_{\beta}}\right).
  \label{eq:two_optimizations}
\end{equation}
The generalization is now straightforward,
\begin{align}
  \lambda_{m+1-\Sigma}(\Gamma)
  &= \underset{\varphi_{\beta}\in I_{\beta}}{\opt_{\sigma_{\beta}}}
  \cdots \underset{\varphi_{2}\in I_{2}}{\opt_{\sigma_{2}}}\,
  \underset{\varphi_{1}\in\pp}{\opt_{\sigma_{1}}}
  \lambda_{m+1}\left(\Gamma_{\varphi_{1},\varphi_{2}\ldots,\varphi_{\beta}}\right),
  \nonumber \\
  &= \underset{\varphi_{\beta}\in I_{\beta}}{\opt_{\sigma_{\beta}}}
  \cdots \underset{\varphi_{2}\in I_{2}}{\opt_{\sigma_{2}}}\,
  \underset{\varphi_{1}\in\pp}{\opt_{\sigma_{1}}}\Lambda\tuple.
  \label{eq:all_optimizations-2}
\end{align}
where $\Sigma=\sigma_{1}+\ldots+\sigma_{\beta}$.

We therefore get that the eigenfunction number $n=m + 1 - \Sigma$ of
$\Gamma$ has $m$ nodal points and therefore
$m+1-\left(\beta_{\Gamma}-\beta\left(Q\right)\right)=m+1-\beta$
nodal domains; the latter follows from \eqref{eq:mu_and_nu_exact}
with $\beta_{\Gamma}=\beta$ and $\beta\left(Q\right)=0$. The nodal
deficiency of this eigenfunction is therefore
$d_n=d_{m+1-\Sigma}=\beta-\Sigma$, which equals the number of
parameters with respect to which we maximize in
(\ref{eq:all_optimizations-2}).

We call expression \eqref{eq:all_optimizations-2} the \emph{mixed
  minimax characterization} of the eigenvalue since some of the
optimizations are minimums and some are maximums.  These, in general,
do not commute, so the minimax cannot be ``unmixed''.
Characterization \eqref{eq:all_optimizations-2} is the main result of
this subsection.

\subsection{Minimax and the Morse index}
\label{sec:morse}

We end the proof of Theorem~\ref{thm:critical_point} by showing that
if the critical point $\ttuple$ is non-degenerate then its Morse index
also equals $\beta-\Sigma$.  In other words, we show that the
deficiency equals the number of negative eigenvalues of the Hessian of
$\Lambda$ at $\ttuple$.

For each value of $\varphi_{2},\ldots,\varphi_{\beta}$
the first optimization of (\ref{eq:all_optimizations-2}),
is achieved for a certain $\varphi_1 =
\psi_{1}\left(\varphi_{2},\ldots,\varphi_{\beta}\right)$, i.e.
\begin{equation*}
  \underset{\varphi_{1}\in\pp}{\opt_{\sigma_{1}}}\Lambda\tuple =
  \Lambda\left(\psi_1(\varphi_{2},\ldots,\varphi_{\beta}),
    \varphi_{2},\ldots,\varphi_{\beta}\right).
\end{equation*}
The function $\psi_1$ satisfies
$\tf_{1}=\psi_{1}\left(\tf_{2},\ldots,\tf_{\beta}\right)$ and
defines a manifold $N_1 \subset \pp^{\beta}$.

On this manifold, for each value of
$\varphi_{3},\ldots,\varphi_{\beta}$, the optimization with respect
to $\varphi_2$ is achieved at $\varphi_2 =
\psi_{2}\left(\varphi_{3},\ldots,\varphi_{\beta}\right)$, which
defines a submanifold $N_2\subset N_1$.  Proceeding in the same
manner, we define a sequence $\left\{ \psi_{j}\right\}
_{j=1}^{\beta}$ of functions and a chain
\begin{equation*}
  \pp^{\beta} \supset N_1 \supset N_2 \ldots \supset N_\beta = \ttuple.
\end{equation*}

Wishing to diagonalize the Hessian at the critical point, $\ttuple$,
we introduce a new set of variables,
$\left(\bar{\varphi}_{1},\ldots,\bar{\varphi}_{\beta}\right)$,
\begin{equation*}
  \bar{\varphi}_{j}=\varphi_{j}-\psi_{j}\left(\varphi_{j+1},\ldots,\varphi_{\beta}\right),
\end{equation*}
for which we get that
\begin{equation*}
  N_{j}\subset\left(0,\ldots0,\bar{\varphi}_{j+1},\ldots,\bar{\varphi}_{\beta}\right).
\end{equation*}
We note that the meaning of the manifolds in the changed variables
remains the same: the extremum on $N_j$ when varying
$\bar{\varphi}_{j+1}$ is achieved when $\bar{\varphi}_j=0$, that is
on $N_{j+1}$.  The extremal property implies that
\begin{equation*}
  \frac{\partial\Lambda}{\partial \bar{\varphi}_{j}}\at_{N_j}=0.
\end{equation*}
Differentiating this identity with respect to $\bar{\varphi}_{k}$
with $k>j$ (so that we remain on $N_j$) we obtain
\begin{equation*}
  \frac{\partial^2\Lambda}{\partial\bar{\varphi}_k
    \partial\bar{\varphi}_j}\at_{N_j} = 0
  \qquad k>j.
\end{equation*}
Since the critical point belongs to all $N_j$ we conclude that its
Hessian is a triangular matrix.  In fact, it is diagonal since
the Hessian is symmetric. The signs of its diagonal entries at the
critical point are known due to the optimization process
\begin{equation}
  \textrm{sign}
  \frac{\partial^{2}\Lambda}{\partial\left(\bar{\varphi}_{j}\right)^{2}}\at_{N_j}
  = 2\sigma_{j}-1,
  \label{eq:sign_of_second_derivative}
\end{equation}
where we used the fact that the critical point is non-degenerate.
The Morse index of the critical point is independent of the choice
of coordinates. We therefore deduce from
\eqref{eq:sign_of_second_derivative} that the Morse index equals to
the number of zeros among $\sigma_j$, that is $\beta-\Sigma$.

\section{Other scenarios}
\label{sec:other_scenarios}

\subsection{Low eigenvalues}

In the discussion so far we restricted our attention to equipartitions
whose parts have no cycles, i.e.\ $\beta_{\Gamma\setminus Q}=0$.
Indeed, Theorems \ref{thm:parametr} and \ref{thm:critical_point}
ignore all other equipartitions. The justification for this is given
by Lemma \ref{lem:estimateN}, which shows that equipartitions with
$\beta_{\Gamma\setminus Q}>0$ do not appear if we restrict ourselves
to high enough eigenvalues.  However, with some extra work it is
possible to extend the treatment to all proper equipartitions.

The parameterization of $\cQ$ (for large enough $m$) was done by
choosing the location of section points, $\left\{ v_{i}\right\}
_{i=1}^{\beta}$, which determine the action of the map $\Phi_m$.  We
now take a different approach, which allows us to relax the
restriction on the value of $m$ at the cost of sacrificing the global structure
of the map $\Phi_m$.  Given an equipartition $Q\in\cQ$, we position
the section points depending on $Q$.  We recursively add a section
point to an edge that contains at least one partition vertex of $Q$ as
long as the new section point does not disconnect the graph.  It is
easy to see that this will result in $k :=
\beta_{\Gamma}-\beta_{\Gamma\setminus Q}$ section points, i.e.\ in
general less than before.  As a result, each cycle of $\Gamma$ will
have a section point if and only if it has a partition vertex of $Q$.
We note that if $\beta_{\Gamma}=\beta_{\Gamma\setminus Q}$ we add no
section point.  One can see that in this case the equipartition $Q$ is
isolated.

We can now define the map $\Phi_m$ from some open set
$D\subset\pp^{k}$ to some neighborhood of $n$-equipartitions
around $Q$. The map acts in the same manner as before (see the
discussion preceeding Theorem \ref{thm:parametr}). For each
$\vec{\varphi}\in D$, $\Phi_m\vec{\left(\varphi\right)}$ is the
equipartition which corresponds to the zeros of the $(m+1)$-th
eigenfunction of $\Gamma_{\vec{\varphi}}$. The validity of the map is
proved in the following theorem.

\begin{theorem}
  \label{thm:loc_parametr}
  Let $Q$ be an $n$-equipartition on a finite connected graph $\Gamma$
  and let $k = \beta_{\Gamma}-\beta_{\Gamma\setminus Q}$.  Denote by
  $\nn(Q)$ the set of proper equipartitions that have the same number
  of partition points on each edge as $Q$.  Then there exists some
 open set $D\subset\pp^{k}$ which is bijectively mapped by
  $\Phi_m$ to $\nn(Q)$.
\end{theorem}

\begin{proof} [Sketch of the proof]
  The proof follows the same procedure as the proof of Theorem
  \ref{thm:parametr}.  For each partition $Q' \in \nn(Q)$ we read off
  the values of $\varphi$ at the section points from the ground states
  on the corresponding subgraph of the partition.  We then reconstruct
  the eigenfunction on $\Gamma_{\vec{\varphi}}$ using the fact that
  whenever there is a zero (and thus matching is required) on a cycle
  of $\Gamma$, this cycle has been cut by a section point (and thus
  matching is possible).

  To see that the obtained eigenfunction of $\Gamma_{\vec{\varphi}}$
  is indeed eigenvalue number $m+1$, we use the fact that
  \begin{equation*}
    \beta_{\Gamma_{\vec{\varphi}}} = \beta_{\Gamma\setminus Q'},
  \end{equation*}
  which turns inequality~\eqref{eq:nodal_zeros_bound_improv} into an
  equation.  The eigenfunction in question has $m$ zeros which makes
  it the eigenfunction number $m+1$.

  Openness of the set $D$ follows from continuity of the position of
  zeros as functions of parameters $\vec{\varphi}$.  That is, we can
  change values of $\vec{\varphi}$ so that the zeros will remain on
  their corresponding edges. Finally, the injectivity of the map is
  verified as in the proof of Theorem \ref{thm:parametr}, part
  \ref{i:bijection}.
\end{proof}

Having established the existence of the map $\Phi_m$ locally around
$Q$, we can use it to define the functional $\Lambda$ on a neighborhood
of $Q$. We then obtain a result identical to Theorem
\ref{thm:critical_point} by following the same proof, as only local
properties of $\Lambda$ were used there.


\subsection{Improper partitions}

In this section we explain the restriction of our results to proper
partitions by giving examples of anomalous behavior of improper ones.
The unifying feature of our examples is the existence of
eigenfunctions vanishing on the vertices of the graph or even on the
entire edges.  These eigenfunctions arise because of the presence of
symmetries in the graphs considered below.

Consider a star graph (see Fig.~\ref{fig:improper}) with 3 edges of
lengths $l-\varepsilon$, $l$, $l$, where $\varepsilon$ is small.  The
vertex conditions at the center are Neumann-Kirchhoff (equation
\eqref{eq:delta_deriv} with $\alpha=0$) and at the outside vertices
are Dirichlet.

We inspect the partitions with one nodal point.  There are three cases to
consider, the point is on the shorter edge, on one of the longer edges
and at the central vertex.  Denote by $\delta$ the distance of the
nodal point to the central vertex.  The corresponding values of
$\Lambda$ are then
\begin{equation*}
  \Lambda_1 = \frac\pi{l-\varepsilon-\delta}, \qquad
  \Lambda_2 = \frac\pi{l-\delta}, \qquad
  \Lambda_3 = \frac\pi{l-\varepsilon}.
\end{equation*}
The infimum of the above values as $\delta$ varies is $\pi/l$, but it is
not achieved since at $\delta=0$ (nodal point on the longer edge) the
functional $\Lambda$ is discontinuous.  We note that this problem
cannot be cured by seeking minimum over the set of partitions with two
parts rather than ``partitions by one point''.  This approach restores
the continuity of $\Lambda$ by removing the offending point from the
domain of definition but the infimum is still not achieved.  Another
anomaly of the above example is that there are no equipartitions with
two parts or, equivalently, with one partition point.

\begin{figure}[ht]
  \begin{centering}
    \hfill
    \begin{minipage}[c]{0.3\columnwidth}%
      \begin{center}
        \includegraphics{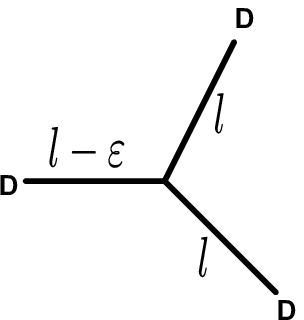}
      \end{center}

      \medskip{}
      \centerline{(a)}
    \end{minipage}
    \hfill{}%
    \begin{minipage}[c]{0.3\columnwidth}%
      \begin{center}
        \includegraphics{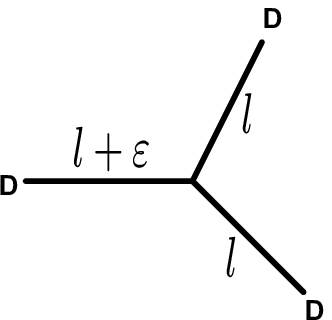}
      \end{center}

      \medskip

      \centerline{(b)}
   \end{minipage}
    \hfill{}
    \begin{minipage}[c]{0.3\columnwidth}%
      \vspace{1cm}
      \begin{center}
        \includegraphics{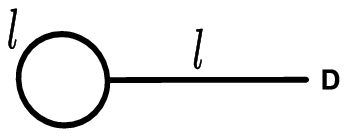}
      \end{center}

      \vspace{1cm}

      \centerline{(c)}
   \end{minipage}
    \hfill{}
  \end{centering}

  \caption{Three examples of graphs with troublesome (improper)
    critical partitions.}

 \label{fig:improper}
\end{figure}

Consider now a slight modification of the above example, a 3-star
graph with edge lengths $l+\varepsilon$, $l$, $l$.  In the set of
partitions by one point the infimum is again not achieved.  The same
applies to the set of all partitions into 2 parts.  However in the set
of partitions into 3 parts the minimum is achieved by the configuration
with a nodal point at the center.  But this is not an equipartition,
violating a would-be analogue of Theorem \ref{thm:minima_of_Lambda}.

In a more sophisticated example of a loop of length $l$ with a edge
of the same length attached to it, a local maximum among
equipartitions with one nodal point is the equipartition with the
point at the central vertex.  The nearby equipartitions are obtained
as the nodal point moves left or right along the loop.\footnote{Such
  partitions have only one part, but they still fit
  Definition~\ref{def:partition}} It can be shown, however, that the
functional $\Lambda$ is not differentiable at the maximum point.

To summarize, the above examples illustrate the necessity of
restricting our attention to the generic case (for example, with
respect to edge length variation) of proper eigenfunctions.

\section{Discussion}

We have investigated the connection between zeros of eigenfunctions of
a quantum graph and ``optimal'' partitions of the said graph.  This
point of view is not new in spectral theory.  In a series of papers
\cite{HHT1,HHT2,HHT3}, which culminated in \cite{HelHofTer_aihp09},
Schr\"odinger operators on domains in $\mathbb{R}^{2}$ were studied
using partitions.  Other energy functionals have also been considered,
see, for example, \cite{ConTerVer_cvpde05,CafLin_jsc07}.  However, it
was the minimizers of the maximum functional,
equation~(\ref{eq:Lambda_def}) that were shown in
\cite{HelHofTer_aihp09} to correspond to certain eigenfunctions, namely
Courant-sharp ones.  However, there is only a finite number of such
eigenfunctions for each domain.

We use the same approach to study the connection between partitions
and eigenfunctions on quantum graphs.  We discover that it is
beneficial to restrict the domain of definition of the functional to
equipartitions, where the maximum functional becomes differentiable.
Upon this restriction, all eigenfunctions of a quantum graph can be
characterized as critical points of the energy functional.  Furthermore,
the Morse index of such a critical point turns out to be equal to the
nodal deficiency of the eigenfunction. In the case that the critical
point is a minimum, we get that the Morse index is zero and therefore
the eigenfunction is Courant-sharp, providing an analogue of the
results obtained by Helffer et al \cite{HelHofTer_aihp09}.

The general nature of our result suggests that analogous theory can be
developed for domains in $\rr^d$ by considering a restriction of the
functional $\Lambda$ to the (now infinite-dimensional) set of
equipartitions.  With the help of the insight gained from the present
results, a theorem analogous to Theorem~\ref{thm:critical_point} has
been established in $\rr^d$ under the assumption that nodal lines
(surfaces) do not intersect \cite{BerKucSmi_prep11}.  A result in the
spirit of Theorem~\ref{thm:critical_point} is also available on
discrete graphs (work in progress).

We finish this discussion with a conjecture.  The nodal deficiency was
determined as the number of maximization in the sequence of $\beta$
operations, see equation~(\ref{eq:all_optimizations-2}).  We
conjecture that asymptotically for large $m$ the choice of minimum or
maximum become ``independent'' and ``random'', in the sense that the
empirical distribution of the deficiencies $d_n$ approaches binomial
distribution with $p=1/2$ and $n=\beta$.

\section{Acknowledgments}

The authors are supported by EPSRC (RB: grant number EP/H028803/1, US:
grant number EP/G021287), NSF (GB: grant DMS-0907968) and BSF (grant
2006065).  The research has been inspired by a talk about the results
of \cite{HelHofTer_aihp09} given by T.~Hoffmann-Ostenhof.  We are
grateful to Peter Kuchment for suggesting a simpler proof for section
\ref{sec:morse}.  GB and HR thank Weizmann Institute, where most of
the work was done, for warm hospitality.  US acknowledges support from
the Minerva Center for Nonlinear Physics, the Einstein (Minerva)
Center at the Weizmann Institute and the Wales Institute of
Mathematical and Computational Sciences (WIMCS).

\appendix

\section{Interlacing theorems for quantum graphs}
\label{sec:interlacing_theorems}

The following proposition, which forms a part of corollary
5.2 from \cite{BerKuc_prep10}, discusses the connection between
manifestations of spectral degeneracy.

\begin{proposition}\label{prop:non_zero_is_simple}
 Let $T$ be a tree with $\delta$-type conditions at its vertices,
  with the exclusion that Dirichlet conditions are not allowed on
  internal vertices. If the eigenvalue $\lambda$ of $T$ has an
  eigenfunction that is non-zero on all internal vertices of $T$ then
  $\lambda$ is simple.
\end{proposition}

We next bring three interlacing theorems from \cite{BerKuc_prep10},
namely theorems 5.1, 5.3 and 5.4.

\begin{theorem}
  \label{thm:interlacing}
  Let $\Graph_{\alpha'}$ be the graph obtained from the graph
  $\Graph_{\alpha}$ by changing the coefficient of the condition at
  vertex $v$ from $\alpha$ to $\alpha'$.  If $-\infty < \alpha <
  \alpha' \leq \infty$ (where $\alpha' = \infty$ corresponds to the
  Dirichlet condition), then
  \begin{equation}
    \label{eq:interlacing_monotone}
    \lambda_n(\Graph_\alpha) \leq \lambda_n(\Graph_{\alpha'}) \leq
    \lambda_{n+1}(\Graph_\alpha).
  \end{equation}
  If the eigenvalue $\lambda_n(\Graph_{\alpha'})$ is simple and it's
  eigenfunction $f$ is such that either $f(v)$ or $\sum f'(v)$ is
  non-zero then the inequalities can be made strict,
  \begin{equation}
    \label{eq:interlacing_strong}
    \lambda_{n}(\Graph_{\alpha}) < \lambda_n(\Graph_{\alpha'})
    < \lambda_{n+1}(\Graph_{\alpha}).
  \end{equation}
\end{theorem}

Next interlacing theorem is

\begin{theorem}
  \label{thm:interlacing_join}
  Let $\Graph$ be a compact (not necessarily connected) graph.  Let
  $v_{0}$ and $v_{1}$ be vertices of the graph $\Graph$ endowed with
  the $\delta$-type conditions, i.e.
  \begin{equation*}
    \begin{cases}
      f\mbox{ is continuous at }v_{j}\mbox{ and}\\
      \sum\limits
      _{e\in\Edges_{v_{j}}}\frac{df}{dx_{e}}(v_{j})=\alpha_{j}f(v_{j}),\qquad
      j=0,1.
    \end{cases}
  \end{equation*}
  Arbitrary self-adjoint conditions are allowed at all other vertices
  of $\Graph$.

  Let $\Graph'$ be the graph obtained from $\Graph$ by gluing the
  vertices $v_{0}$ and $v_{1}$ together into a single vertex $v$, so
  that $\Edges_{v}=\Edges_{v_{0}}\cup\Edges_{v_{1}}$, and endowed with
  the $\delta$-type condition
  \begin{equation}
    \sum\limits_{e\in\Edges_{v}}\frac{df}{dx_{e}}(v)=(\alpha_{0}+\alpha_{1})f(v).
    \label{eq:glue_conditions}
  \end{equation}

  Then the eigenvalues of the two graphs satisfy the
  inequalities
  \begin{equation}
    \lambda_{n}(\Graph)\leq\lambda_{n}(\Graph')\leq\lambda_{m+1}(\Graph).
    \label{eq:interlacing_join}
  \end{equation}
\end{theorem}

In the current manuscript we apply the above theorem with
$\alpha_{0}=-\alpha_{1}$ and a slight adaptation of
\eqref{eq:interlacing_strong}:
\begin{equation}
  \lambda_{n}(\Graph')\leq\lambda_{m+1}(\Graph)\leq\lambda_{m+1}(\Graph').
  \label{eq:our_interlacing_join}
\end{equation}

Repeated applications of the theorem above gives the following result,
which is a quote of theorem 4.6 from \cite{BerKuc_prep10}.

\begin{theorem}
  \label{thm:glue_many}
  Let the graph $\Gamma'$ be obtained from $\Gamma$ by $k$
  identifications, for example by gluing vertices
  $v_{0},v_{1},\ldots,v_{k}$ into one, or pairwise gluing of $k$ pairs
  of vertices. Each identification results also in adding the
  $\alpha_{j}$ parameters in the vertex $\delta$-type conditions, as
  in \eqref{eq:glue_conditions}. Then
  \begin{equation*}
    \lambda_{n}(\Graph)\leq\lambda_{n}(\Graph')\leq\lambda_{n+k}(\Graph).
  \end{equation*}
\end{theorem}

\bibliographystyle{abbrv}

\bibliography{bk_bibl_robin,additional,Nodal_Domains_general,Nodal_Domains_of_Quantum_Graphs,Qunatum_Graphs,partitions}

\end{document}